\pdfoutput=1
\documentclass[
  aps, pre, reprint, amsmath, amssymb,
  superscriptaddress, nofootinbib, floatfix
]{revtex4-2}

\usepackage{bm}
\usepackage{pgfplots}
\usepgfplotslibrary{groupplots}
\pgfplotsset{compat=1.18}
\usepackage{mathtools}
\usepackage{microtype}
\usepackage{amsthm}

\theoremstyle{plain}
\newtheorem{proposition}{Proposition}
\newtheorem{corollary}[proposition]{Corollary}
\newtheorem{conjecture}{Conjecture}
\theoremstyle{definition}
\newtheorem{definition}{Definition}
\theoremstyle{remark}
\newtheorem*{remark}{Remark}

\allowdisplaybreaks
\newcommand{\kT}{k_{B}T}
\newcommand{\Hp}{H_{\perp}}
\newcommand{\Hpk}[1]{H_{\perp,#1}}
\newcommand{\Dp}{D_{\perp}}
\newcommand{\Dpar}{D_{\parallel}}
\newcommand{\Gsel}{\mathcal{G}}
\newcommand{\Fisher}{\mathcal{I}}
\newcommand{\dd}{\mathrm{d}}
\newcommand{\spec}{\operatorname{spec}}
\newcommand{\Vd}{V^{\ddagger}}
\newcommand{\fnc}{f_{\rm nc}}
\newcommand{\Kfr}{K}

\begin{document}

\title{Channel selection at identically vanishing dissipation difference:\\
       isolating the frenetic sector of the overdamped path measure}

\author{Shlomo Segal}
\noaffiliation

%=====================================================================
\begin{abstract}
Transition-state and Kramers--Langer theory determine reaction channel selection
from data at a single point: the saddle. We show this is insufficient in a
precise, constructive sense, using the standard split of the Onsager--Machlup
path weight into a time-antisymmetric part fixed by the entropy flux to the
medium and a time-symmetric part, the frenesy or dynamical activity. We first
settle when a dissipation-based criterion can distinguish two \emph{forward}
histories at all. For overdamped dynamics with drift $\beta D(-\nabla V+\fnc)$
the answer is exact and purely topological: the entropy difference between two
histories sharing endpoints is $\beta$ times the circulation of $\fnc$ around the
loop they enclose, and it vanishes identically for every gradient flow. We then
construct a two-channel gradient landscape whose channels share both endpoints,
whose barrier heights are equal identically, and whose saddle Hessians are the
same matrix. The entropy sector is therefore switched off by construction, all
local saddle-point theories predict a $50{:}50$ split exactly, and the branching
ratio isolates the frenetic sector. Direct Langevin simulation gives
$P(+)=0.4568\pm0.0012$. We identify a covariant geometric contribution to
pathway selection within the frenetic sector, $\tfrac12\ln\det(\Dp\Hp)$, show
that the divergence-of-drift observable previously proposed for this role is not
invariant under nonlinear coordinate change whereas this one is, and find that it
reproduces the measured branching to within $0.5\sigma$ with no fitted
parameters, tracking a one-parameter family with a passing null control. Because
selection depends only on a ratio of determinants, stiff spectator modes cancel
identically. Finally we give a parameter-free prediction for the complementary
sector: a solenoidal force of strength $\varepsilon$ must contribute exactly
$2\beta\varepsilon/\pi$ to $\ln[P(+)/P(-)]$. These results establish an exactly
controlled counterexample to the universal sufficiency of local saddle-point
information for reaction-channel selection.
\end{abstract}

\maketitle

%=====================================================================
\section{Introduction}

\subsection{The problem}

The dominant framework for predicting which of several competing pathways a
system takes is transition-state theory (TST): compare activation free energies,
and the lower barrier wins, exponentially. Where dynamical corrections are
needed, Kramers--Langer theory~\cite{Kramers1940,Langer1969} supplies a prefactor
built from the Hessian of the potential at the saddle and the friction tensor
there.

Both are \emph{local at the saddle}. Everything the theory knows about the
landscape sits in one point per channel.

This is usually adequate, because the exponential dominates: a barrier difference
of a few $\kT$ swamps any prefactor of order unity. But there is a class of
problems where it fails by construction---those in which the barriers are
degenerate. Then the exponential carries no information at all, the prefactor is
the entire signal, and a theory evaluating the prefactor at a single point can be
shown to give the wrong answer. Post-transition-state bifurcations are the
extreme case~\cite{Ess2008,Hare2017}: one saddle, two products, activation free
energy difference identically zero, saddle Hessians the same object---yet
branching ratios that are generically not $1{:}1$ and are experimentally
reproducible.

\subsection{Two sectors of the path measure}
\label{sec:twosectors}

The organising observation of this paper is not new, and we do not present it as
new. For a Markov diffusion the weight of a path splits into a part odd under
time reversal, fixed by the entropy flux to the medium, and a part even under
time reversal, the \emph{frenesy} or dynamical
activity~\cite{Maes2020,Baiesi2009,MaesNetocny2008,Seifert2012}. Schematically,
\begin{equation}
  P[\gamma]\;\propto\;\exp\!\bigl(\tfrac12\sigma[\gamma]-\Kfr[\gamma]\bigr).
  \label{eq:split0}
\end{equation}
Fluctuation theorems of Crooks type~\cite{Crooks1999,Lebowitz1999} constrain the
first term only. The second is a separate object, and in a large class of
problems it is the only one that carries any signal.

What this paper does is to make that last clause quantitative, and then to build
a system in which it is realised exactly. Section~\ref{sec:when} answers, for
overdamped dynamics, the question of \emph{when} the entropy sector can
distinguish two forward histories: the answer is a circulation, and it is zero
for every gradient flow. Section~\ref{sec:cov} identifies the geometric component
of the frenesy and shows it is a coordinate scalar, correcting an observable
proposed earlier for this role. Section~\ref{sec:numerics} exhibits a landscape
in which the entropy sector vanishes identically and measures what is left.

\subsection{Relation to two earlier preprints}
\label{sec:intro-rel}

This work stands in a different relation to each of two earlier preprints by the
present author, and we distinguish the two relations sharply.

\paragraph*{Reference~\cite{SegalGeom}: a repair.}
That preprint proposed a local observable
$\chi^{*}=\int(\nabla E\cdot A)^{2}/(\epsilon_{0}+|\nabla\!\cdot\!A|)\,\rho_{\rm ss}$,
derived from a Wasserstein gradient-flow formulation~\cite{JKO1998}, and argued
that it distinguishes kinetically distinct channels that free-energy analysis
cannot separate. The target problem is real and the geometric intuition---that
structure beyond the saddle governs channel selection---is, we believe, correct,
and it is placed on a rigorous footing here. What requires revision is the
specific observable. There are three reasons a covariant reformulation is needed,
in increasing order of importance. First, the
numerator is $|\nabla E|^{4}/\gamma^{2}$, a purely first-order quantity: the
geometry resides entirely in the denominator. Second, that denominator is a trace,
$|\nabla\!\cdot\! A|=|\operatorname{Tr}H|/\gamma$, so two channels with equal
trace and different spectra are indistinguishable---precisely the case
constructed in Sec.~\ref{sec:landscape}. Third, and decisively,
$\nabla\!\cdot\! A$ is \emph{not a scalar} under nonlinear coordinate change;
Sec.~\ref{sec:invtest} exhibits a point in a physical landscape at which a smooth
relabelling reverses its sign while altering no observable. An observable that
changes sign under relabelling cannot select anything. Section~\ref{sec:cov}
supplies the covariant form and identifies it with the geometric component
of $\Kfr$, realising the intuition of Ref.~\cite{SegalGeom} in an invariant way.

\paragraph*{Reference~\cite{SegalOoL}: a qualification, and a provenance.}
That preprint proposed, in the tradition of dissipation-based accounts of
self-organisation~\cite{England2013}, that non-equilibrium histories carry a
probabilistic bias toward greater integrated dissipation, and inferred the
selection of template-directed replication over simple autocatalysis. Its central step, which
it labels an ansatz, is
$P(x_{1})/P(x_{2})\approx e^{(\sigma_{1}-\sigma_{2})/2}$ for two distinct forward
histories.

Three things can now be said that the preprint does not say, and we regard the
first two as favourable to it.

\emph{The ansatz is not a guess, and its coefficient is exact.}
Equation~\eqref{eq:split} below shows that it is precisely the
$\Delta \Kfr\to0$ limit of the standard time-reversal decomposition of the path
action. The factor $\tfrac12$, the sign, and the direction of the bias are all
correct, and they follow from the structure of the path measure rather than from
analogy. What the preprint presents as a modelling assumption is in fact a
controlled approximation with a single identifiable remainder.

\emph{Its validity is governed by one measurable ratio.}
Section~\ref{sec:disc-ool} reports an exact enumeration study establishing that
the accuracy of the ansatz is controlled not by the strength of the driving---as
one might expect, and as the preprint implicitly assumes---but by
$R=\mathrm{std}(\Delta \Kfr)/\mathrm{std}(\Delta\sigma/2)$, the spread of the
neglected term relative to the retained one. This is a usable criterion, and it
was not previously available.

\emph{The attribution is wrong and the asymptotic conclusion is unproven.}
The preprint derives the ansatz from the Crooks
theorem~\cite{Crooks1999}; that route does not exist, since Crooks compares a
path with its own time reverse. The correct provenance is
Eq.~\eqref{eq:split}. Separately, the doubly-exponential selection ratio does not
follow, because the neglected term is extensive in the same quantity as the
retained one; Sec.~\ref{sec:disc-ool} quantifies this. We emphasise that this is
a statement of non-proof, not of refutation: our calculations leave the
qualitative conclusion open and identify precisely the quantity that decides it.

Nothing below contains a population variable, a replicator, an error threshold,
or a chemical network. What this paper supplies to that line of work is the
correct decomposition, an exact criterion for when a dissipation-based mechanism
can discriminate two forward histories at all (Sec.~\ref{sec:when}), a validity
criterion for the ansatz, and a measurement of the neglected sector in a
controlled setting.

\subsection{Status of claims}
\label{sec:status}

Because this paper mixes theorems, simulations and conjectures, we sort them
explicitly. Table~\ref{tab:status} is the authoritative statement of what is
being asserted; the prose elsewhere should be read against it.

\begin{table*}[t]
\caption{\label{tab:status}%
Status of the claims of this paper. ``Proved'' means proved here or elementary;
``supported'' means established numerically within the stated statistics for the
specific system studied; ``open'' means neither established nor refuted by
anything computed here; ``conjectured'' means stated and not established.}
\begin{ruledtabular}
\begin{tabular}{lc}
Claim & Status\\
\colrule
Time-reversal split, Eq.~\eqref{eq:split}, in the stated convention & known\\
Circulation identity, Prop.~\ref{prop:circ} & proved\\
$\Delta\sigma\equiv0$ for the benchmark, Cor.~\ref{cor:zero} & proved\\
$\nabla\!\cdot b$ not a scalar; $\det(DH)$ is, Prop.~\ref{prop:cov} & proved\\
$\chi^{*}$ of~\cite{SegalGeom} not coordinate-invariant, Prop.~\ref{prop:cov} & proved\\
Insufficiency of saddle data, Prop.~\ref{prop:insuff} & proved\\
Spectator cancellation, Prop.~\ref{prop:spectator} & proved\\
$\Phi$ predicts the measured $\Delta \Kfr$ & supported\\
$\Phi$ tracks the $\delta_{0}$ family & supported\\
Slope $2\beta/\pi$ in the solenoidal test, Sec.~\ref{sec:eps} & predicted\\
Ansatz of~\cite{SegalOoL} exact as $\Delta \Kfr\to0$ & proved\\
Validity of that ansatz controlled by $R$, Table~\ref{tab:ansatz} & supported\\
Doubly-exponential conclusion of~\cite{SegalOoL} & open\\
Two-tier functional, Conj.~\ref{conj:twotier} & conjectured\\
Frenetic reading of the error threshold, Conj.~\ref{conj:eigen} & conjectured\\
Anything about replicators or abiogenesis & not addressed\\
\end{tabular}
\end{ruledtabular}
\end{table*}

\subsection{What is and is not new here}
\label{sec:novelty}

We state this plainly, because several ingredients are standard.

The entropy/frenesy split is standard~\cite{Maes2020,Baiesi2009,MaesNetocny2008}
and we claim no part of it. The circulation identity of
Prop.~\ref{prop:circ} is Stokes' theorem; its content is in the question it
answers, not in its proof. The effective potential $V+\tfrac{\kT}{2}\ln\det\Hp$
is the anisotropic generalisation of the Fick--Jacobs entropy
potential~\cite{Jacobs1967,Zwanzig1992,Reguera2001}, and channel-width modulation
of flux is standard entropic-barrier physics. That the Hessian is a $(0,2)$-tensor
and the diffusion tensor a $(2,0)$-tensor, so that their product is spectrally
invariant, is elementary.

What we do claim, in descending order of what we believe survives scrutiny:

\emph{(1) A benchmark in which the entropy sector is switched off identically}
(Sec.~\ref{sec:landscape}, Cor.~\ref{cor:zero}): a smooth two-dimensional
gradient landscape whose two channels share both endpoints, whose barrier heights
are equal as an algebraic consequence of the functional form rather than to
numerical tolerance, and whose saddle Hessians are the same matrix. The branching
ratio therefore isolates the frenetic contribution by construction. We are not
aware of an existing benchmark this sharp and offer it as a test system for any
method claiming to predict channel selection.

\emph{(2) Identification of the covariant geometric component of the frenesy}
(Sec.~\ref{sec:cov}) and the correction of Ref.~\cite{SegalGeom} that follows.
The mathematics is standard; the identification, and the demonstration that the
earlier observable fails a coordinate-change test on a physical landscape, are
not.

\emph{(3) A spectator-cancellation result} (Prop.~\ref{prop:spectator}), reducing
the effective dimensionality from $3N-6$ to the number of modes whose stiffness
differs between channels.

\emph{(4) A parameter-free prediction for the complementary sector}
(Sec.~\ref{sec:eps}): a solenoidal perturbation of known circulation must shift
the log-branching by an amount fixed entirely by the enclosed area. No constant
is fitted. It has not been measured.

We claim no unification of Refs.~\cite{SegalGeom} and~\cite{SegalOoL}, and no
result bearing on the origin of life.

%=====================================================================
\section{The two sectors, and when each carries signal}
\label{sec:when}

\subsection{Dynamics and path measure}

Consider overdamped Langevin dynamics on $\mathbb{R}^{n}$ with potential $V$,
temperature $\kT=\beta^{-1}$, isotropic constant diffusivity $D$, and a possibly
non-conservative force $\fnc$,
\begin{equation}
  \dot x=b(x)+\sqrt{2D}\,\xi(t),
  \qquad
  b=\beta D\bigl(-\nabla V+\fnc\bigr),
  \label{eq:langevin}
\end{equation}
with $\xi$ standard white noise. In the midpoint (Stratonovich) convention the
Onsager--Machlup action~\cite{OnsagerMachlup1953} of a path $\gamma$ on $[0,T]$ is
\begin{equation}
  S[\gamma]=\frac{1}{4D}\int_{0}^{T}\!\!\bigl|\dot\gamma-b(\gamma)\bigr|^{2}\dd t
            +\frac12\int_{0}^{T}\!\!\nabla\!\cdot b(\gamma)\,\dd t,
  \label{eq:OM}
\end{equation}
with $P[\gamma]\propto e^{-S[\gamma]}$.

\subsection{Time-reversal split}

Let $\bar\gamma(t)=\gamma(T-t)$. Expanding the square in Eq.~\eqref{eq:OM},
only the cross term is odd under reversal, so
\begin{equation}
  S[\bar\gamma]-S[\gamma]=\frac{1}{D}\int_{\gamma}b\cdot\dd\gamma
  \;\equiv\;\sigma[\gamma],
  \label{eq:sigmadef}
\end{equation}
$\sigma$ being the entropy flux to the medium. Splitting $S$ into its odd and
even parts under reversal gives
\begin{equation}
  \boxed{\;
  P[\gamma]\propto\exp\!\bigl(\tfrac12\sigma[\gamma]-\Kfr[\gamma]\bigr),
  \qquad
  \Kfr\equiv\tfrac12\bigl(S[\gamma]+S[\bar\gamma]\bigr),\;}
  \label{eq:split}
\end{equation}
explicitly
\begin{equation}
  \Kfr[\gamma]=\frac{1}{4D}\!\int\!|\dot\gamma|^{2}\dd t
              +\frac{1}{4D}\!\int\!|b|^{2}\dd t
              +\frac12\!\int\!\nabla\!\cdot b\,\dd t .
  \label{eq:frenesy}
\end{equation}

\begin{remark}[on conventions, and on what is convention-free]
Equation~\eqref{eq:split} is a rearrangement of Eq.~\eqref{eq:OM}, not a law of
nature, and Eq.~\eqref{eq:OM} is convention-dependent: a different discretisation
of the stochastic integral moves terms between the quadratic part and the
Jacobian, and hence redistributes them within $\Kfr$. What is
convention-independent is the \emph{odd} part, Eq.~\eqref{eq:sigmadef}, which is
fixed by local detailed balance and reproduces the Crooks
relation~\cite{Crooks1999}. All statements below concerning $\sigma$ are therefore
convention-free; statements concerning $\Kfr$ are to be read within the stated
convention. This split is standard~\cite{Maes2020,Baiesi2009,MaesNetocny2008} and
we use it only as an organising device.
\end{remark}

\subsection{When can dissipation distinguish two forward histories?}

Fluctuation theorems compare a path with its own reverse. Selection arguments
require instead a comparison of two \emph{different forward} histories. The
following settles when the entropy sector can supply such a comparison at all.

\begin{proposition}[circulation identity]
\label{prop:circ}
Let $\gamma_{1},\gamma_{2}$ be two histories of the
dynamics~\eqref{eq:langevin} with common initial point and common final point,
and let $\Sigma$ be any surface bounded by the closed loop
$\gamma_{1}\circ\bar\gamma_{2}$. Then
\begin{equation}
  \sigma[\gamma_{1}]-\sigma[\gamma_{2}]
  =\beta\oint_{\gamma_{1}\circ\bar\gamma_{2}}\!\!\fnc\cdot\dd x
  =\beta\!\int_{\Sigma}(\nabla\times \fnc)\cdot\dd\bm{\Sigma}.
  \label{eq:circ}
\end{equation}
\end{proposition}

\begin{proof}
From Eqs.~\eqref{eq:langevin} and~\eqref{eq:sigmadef},
$\sigma[\gamma]=\beta\int_{\gamma}(-\nabla V+\fnc)\cdot\dd\gamma
 =-\beta\,\Delta V+\beta\int_{\gamma}\fnc\cdot\dd\gamma$.
The first term depends only on the endpoints, which the two histories share, and
cancels in the difference. The remainder is a line integral over the closed loop;
Stokes' theorem gives the surface form.
\end{proof}

\begin{corollary}[gradient flows carry no entropic signal]
\label{cor:zero}
If $\fnc\equiv0$ then $\sigma[\gamma_{1}]=\sigma[\gamma_{2}]$ identically for
\emph{every} pair of histories sharing endpoints, and
Eq.~\eqref{eq:split} reduces to
$P[\gamma_{1}]/P[\gamma_{2}]=e^{-(\Kfr_{1}-\Kfr_{2})}$.
\end{corollary}

\begin{remark}
Proposition~\ref{prop:circ} is Stokes' theorem and we claim no mathematical
novelty for it. Its use here is to convert a vague question---whether
``dissipation-driven selection'' can discriminate between pathways---into a sharp
one with a computable answer. Three consequences are worth stating. First, in a
gradient flow the answer is no, identically, at all temperatures and for all
observation times. Second, in a driven system the answer depends only on the
enclosed circulation, so two histories that enclose no flux are entropically
indistinguishable however much each of them dissipates. Third, the magnitude of
the entropic signal is a geometric quantity---an enclosed area weighted by
$\nabla\times \fnc$---and is therefore predictable without simulation. We exploit
the third point in Sec.~\ref{sec:eps}.
\end{remark}

\subsection{Degeneracy conditions}
\label{sec:degen}

Let channels $\gamma_{1},\gamma_{2}$ connect the same reactant and product
basins. Write $\Delta\Vd_{k}$ for the barrier along channel $k$ and
$H^{\ddagger}_{k}$ for the Hessian at its saddle.

\begin{definition}[rate degeneracy]
The channels are \emph{rate-degenerate} if $\Delta\Vd_{1}=\Delta\Vd_{2}$.
\end{definition}

\begin{definition}[Langer degeneracy]
They are \emph{Langer-degenerate} if in addition
$\spec(DH^{\ddagger}_{1})=\spec(DH^{\ddagger}_{2})$.
\end{definition}

\begin{definition}[entropic degeneracy]
They are \emph{entropically degenerate} if $\sigma[\gamma_{1}]=\sigma[\gamma_{2}]$.
By Cor.~\ref{cor:zero} every pair of channels of a gradient flow sharing
endpoints is entropically degenerate.
\end{definition}

Under Langer degeneracy both the TST rate and the Langer-corrected rate coincide
for the two channels, and every method whose output is a function of barrier
height and saddle curvature predicts $1{:}1$---not approximately, but as a
statement about what those theories can express. Under simultaneous Langer and
entropic degeneracy the branching ratio is, by Eq.~\eqref{eq:split}, a pure
measurement of $\Delta \Kfr$. Section~\ref{sec:numerics} realises both conditions
at once.

%=====================================================================
\section{The geometric component of the frenesy}
\label{sec:cov}

\subsection{Tube reduction}

Let $\gamma$ be a smooth curve parameterised by arclength $s$ with unit tangent
$\hat t(s)$. At each $s$ let $\Hp(s)$ be the Hessian of $V$ restricted to the
hyperplane orthogonal to $\hat t$, and $\Dp(s)$ the corresponding block of $D$.
Assume transverse relaxation is fast compared with progress along
$\gamma$---the adiabatic, or Fick--Jacobs, condition; see Sec.~\ref{sec:limits}
for its cost. Integrating out the transverse Gaussian fluctuations,
\begin{equation}
  Z_{\perp}(s)=\int\!\dd^{\,n-1}u\;e^{-\beta u^{\!\top}\Hp(s)u/2}
  =\frac{(2\pi\kT)^{(n-1)/2}}{\sqrt{\det\Hp(s)}},
  \label{eq:Zperp}
\end{equation}
which yields an effective one-dimensional free energy along the tube,
\begin{equation}
  \boxed{\;
  \begin{gathered}
    \mathcal{F}(s)=V(\gamma(s))+\Phi(s),\\[2pt]
    \Phi(s)\equiv\frac{\kT}{2}\ln\det\bigl[\tau_{0}\beta\,\Dp(s)\Hp(s)\bigr].
  \end{gathered}\;}
  \label{eq:Fbox}
\end{equation}
Here $\tau_{0}$ is an arbitrary reference time that cancels from every ratio
below. We call $\Phi$ the \emph{geometric potential}.

Two features locate $\Phi$ within the split~\eqref{eq:split}. It enters
$\mathcal{F}$ with a factor $\kT$, so in $\beta\mathcal{F}$ it carries no power
of $\beta$: it is a prefactor-order object. And it descends from the Jacobian
term $\tfrac12\!\int\!\nabla\!\cdot b\,\dd t$ of Eq.~\eqref{eq:frenesy}, which is
time-symmetric. $\Phi$ is thus a component of the frenesy, not of the entropy
flux---the transverse-geometric component, obtained after tube reduction. We do
not claim it exhausts $\Kfr$: the kinetic and $|b|^{2}$ terms of
Eq.~\eqref{eq:frenesy} contribute as well, and are not evaluated here.

\subsection{Invariance}

\begin{proposition}[covariance]
\label{prop:cov}
Let $u=\varphi(x)$ be a diffeomorphism with Jacobian
$J=\partial\varphi/\partial x$. Then \emph{(i)} $\nabla\!\cdot b$ is not
invariant, and \emph{(ii)} $\det(DH)$ is invariant.
\end{proposition}

\begin{proof}[Proof of (ii)]
$V$ is a scalar, so its Hessian transforms as a $(0,2)$-tensor on the relevant
subspace, $H\mapsto J^{-\top}HJ^{-1}$; the diffusion tensor is a $(2,0)$-tensor,
$D\mapsto JDJ^{\top}$. Hence
\begin{equation}
  DH\;\longmapsto\;JDJ^{\top}J^{-\top}HJ^{-1}=J\,(DH)\,J^{-1},
\end{equation}
a similarity transformation, whose determinant---indeed whose entire
spectrum---is invariant. Claim (i) is established by the explicit counterexample
of Sec.~\ref{sec:invtest}.
\end{proof}

So $DH$ is a $(1,1)$-tensor while $D$ and $H$ separately are not, and $\Phi$ is a
genuine scalar field. The observable of Ref.~\cite{SegalGeom} is the trace-like
non-invariant cousin of this determinant; the repair is to contract the Hessian
with the diffusion tensor \emph{before} taking the determinant, not after, and to
use the full spectrum rather than the trace.

\begin{remark}
We prove invariance, not uniqueness. Any spectral invariant of $DH$ is equally
invariant. Our claim for $\ln\det$ is not that it is the only invariant available
but that it is the one that \emph{appears}: it is what the Gaussian transverse
integral~\eqref{eq:Zperp} produces, so it enters with the correct coefficient.
Invariance is necessary; derivation from the path measure selects this particular
invariant.
\end{remark}

\subsection{The selection index}

For a channel $\gamma_{k}$ with longitudinal diffusivity $\Dpar$ define
\begin{equation}
  \dd\mu_{k}(s)=\frac{e^{\beta\mathcal{F}_{k}(s)}}{\Dpar(s)}\,\dd s,
  \qquad
  R_{k}=\int_{\gamma_{k}}\dd\mu_{k}(s),
  \label{eq:R}
\end{equation}
and for two channels sharing endpoints
\begin{equation}
  \boxed{\;
  \frac{P_{1}}{P_{2}}=\frac{R_{2}}{R_{1}}=e^{\Gsel},
  \qquad
  \Gsel\equiv\ln R_{2}-\ln R_{1}.\;}
  \label{eq:G}
\end{equation}
We call $\Gsel$ the \emph{geometric selection index}. By Cor.~\ref{cor:zero}, for
a gradient flow $\Gsel$ is a prediction for $-\Delta \Kfr$ at the level of
channel classes, where we define the class-level frenetic difference
\begin{equation}
  \Delta \Kfr\;\equiv\;-\ln\bigl[P_{1}/P_{2}\bigr].
  \label{eq:dKdef}
\end{equation}
Equation~\eqref{eq:dKdef} refers to the summed weight of each channel class, not
to individual paths; the two coincide only under a dominant-path approximation,
and we use the class-level object throughout.

\begin{proposition}[insufficiency of saddle-point data]
\label{prop:insuff}
Under Langer degeneracy, TST and Langer theory both give $P_{1}/P_{2}=1$, whereas
$\Gsel\neq0$ generically. Writing $\mathcal{F}_{k}=V_{k}+\Phi_{k}$ with
$V_{1},V_{2}$ sharing a common maximum value, $\Gsel$ vanishes only if
$\Phi_{1}=\Phi_{2}$ almost everywhere on the integration region, not merely at
the saddle.
\end{proposition}

The content is that the two theories differ by a \emph{functional} of the whole
channel versus a \emph{function} of one point, and no amount of saddle-point
refinement closes the gap.

%=====================================================================
\section{Information-geometric identity}
\label{sec:fisher}

The local transverse fluctuation ensemble at arclength $s$ is Gaussian with
covariance $\Sigma_{\perp}(s)=\kT\,\Hp(s)^{-1}$. Its Fisher information matrix
with respect to translations of the mean is
$\Fisher(s)=\Sigma_{\perp}(s)^{-1}=\beta\Hp(s)$, so
\begin{equation}
  \boxed{\;\Phi(s)=\frac{\kT}{2}
  \ln\det\bigl[\tau_{0}\,\Dp(s)\,\Fisher(s)\bigr].\;}
  \label{eq:fisher}
\end{equation}
The geometric component of the frenesy is a log-determinant of Fisher information
contracted with the diffusion tensor. This is an identity obtained by
substitution, not a theorem, and we present it as a change of vocabulary. Two
remarks. The product $\Dp\Fisher$ is dimensionally an inverse relaxation-time
matrix, so $\Phi$ is $\tfrac{\kT}{2}$ times the log-volume of the transverse
relaxation-rate spectrum: channels that relax transversely fast are penalised,
channels that are informationally loose are favoured---which is the frenetic
statement, since fast transverse relaxation is high dynamical activity. And the
thermodynamic uncertainty relation~\cite{Barato2015} bounds precision per unit
dissipation; since $\Gsel$ is channel-discriminating information available at
\emph{zero} dissipation difference, a bound of that type cannot constrain it, and
we note this only to forestall the expectation that it should.

%=====================================================================
\section{Numerical test}
\label{sec:numerics}

\subsection{A landscape with both degeneracies}
\label{sec:landscape}

We construct a two-dimensional \emph{gradient} landscape, $\fnc\equiv0$, with two
channels connecting $x=0$ to $x=1$, symmetric under $y\to-y$ except for a
controlled asymmetry in transverse stiffness that vanishes identically at the
saddle:
\begin{align}
  V(x,y) &= U_{0}\sin^{2}(\pi x)+a(x,y)\bigl[y^{2}-c(x)^{2}\bigr]^{2},
  \label{eq:V}\\[2pt]
  c(x) &= A\sin(\pi x),\\[2pt]
  a(x,y) &= a_{0}\bigl[1+\delta(x)\tanh(y/w)\bigr],\\[2pt]
  \delta(x) &= \delta_{0}\bigl[1-e^{-(x-1/2)^{2}/2\sigma_{w}^{2}}\bigr]\sin^{2}(\pi x).
\end{align}
Parameters are $U_{0}=2.5$, $A=0.5$, $a_{0}=25$, $w=0.15$, $\sigma_{w}=0.06$, and
$\delta_{0}=0.6$ unless stated, in units $\kT=1$, $D=1$.

The construction has four exact properties, each by design.

\emph{(i) Identical barriers.} The channel bottoms are exactly
$U_{0}\sin^{2}(\pi x)$ for both channels, because the quartic factor vanishes
identically at $y=\pm c(x)$ regardless of $a$. Barrier heights are equal
\emph{exactly}, not to numerical tolerance.

\emph{(ii) Identical saddle Hessians.} $\delta(1/2)=0$ and $\delta'(1/2)=0$, so
the two saddles at $(1/2,\pm A)$ have the same Hessian
(Table~\ref{tab:saddle}). The system is Langer-degenerate.

\emph{(iii) Entropic degeneracy.} The dynamics is a gradient flow and
$c(0)=c(1)=0$, so both channels run from $(0,0)$ to $(1,0)$. By
Cor.~\ref{cor:zero}, $\sigma[\gamma_{+}]=\sigma[\gamma_{-}]$
\emph{identically}---at every temperature, for every observation time, and for
every pair of individual trajectories, not merely on average. The entropy sector
of Eq.~\eqref{eq:split} is switched off by construction.

\emph{(iv) Geometric contrast away from the saddle.} $\delta(x)$ grows to
$\delta_{0}\sin^{2}(\pi x)$ away from the saddle region, so the channels differ
in transverse stiffness everywhere except at the saddle itself
(Table~\ref{tab:width}, Fig.~\ref{fig:main}b).

Properties \emph{(i)--(iii)} together mean that the measured branching ratio is a
determination of the class-level frenetic difference $\Delta \Kfr$ of
Eq.~\eqref{eq:dKdef}, with no entropic contamination and with every
saddle-based theory pinned to $1{:}1$.

\begin{table*}[t]
\caption{\label{tab:saddle}%
Numerical verification of Langer degeneracy for the landscape~\eqref{eq:V}. The
saddle data are identical to every printed digit, so TST and Langer theory
predict $P(+)=0.5000$ exactly.}
\begin{ruledtabular}
\begin{tabular}{lcccc}
Channel & Saddle & $\Vd$ & $\spec(H^{\ddagger})$ & $\det H^{\ddagger}$\\
\colrule
$+$ & $(0.5,+0.5)$ & $2.5000000000$ & $\{-49.34802191,\;50.00000001\}$ & $-2467.40109609$\\
$-$ & $(0.5,-0.5)$ & $2.5000000000$ & $\{-49.34802191,\;50.00000001\}$ & $-2467.40109609$\\
\end{tabular}
\end{ruledtabular}
\end{table*}

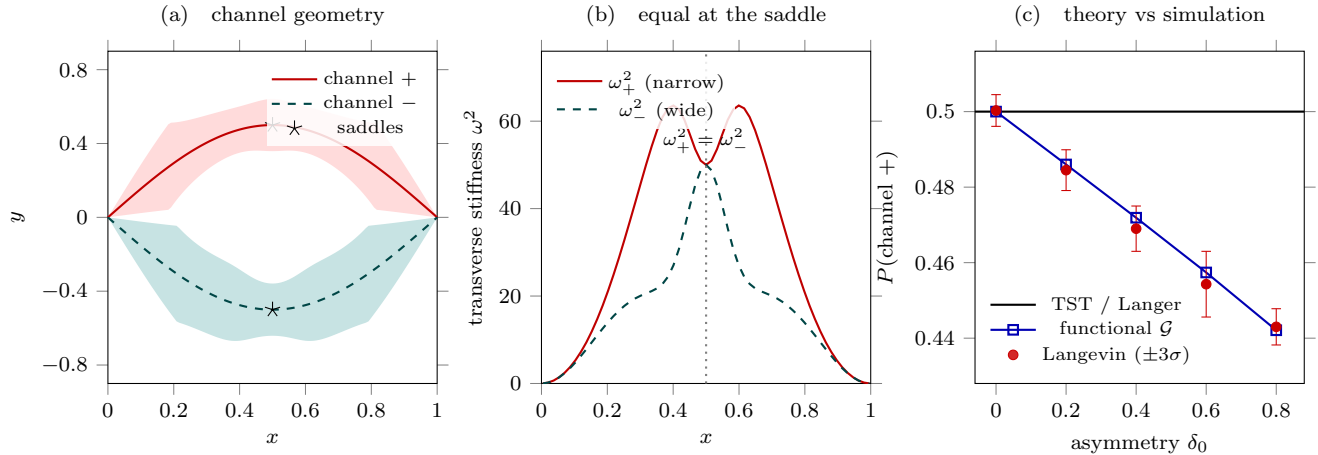
\begin{figure*}[t]
\centering
\begin{tikzpicture}
\begin{groupplot}[
  group style={group size=3 by 1, horizontal sep=1.38cm},
  width=0.243\textwidth, height=4.4cm, scale only axis,
  tick align=outside,
  tick label style={font=\footnotesize},
  label style={font=\footnotesize},
  title style={font=\footnotesize},
  legend style={font=\scriptsize, draw=none, fill=white, fill opacity=0.85,
                text opacity=1, inner sep=1.5pt, row sep=-1pt},
]

%---------------- (a) channel geometry ----------------
\nextgroupplot[
  xlabel={$x$}, ylabel={$y$},
  xmin=0, xmax=1, ymin=-0.9, ymax=0.9,
  ytick={-0.8,-0.4,0,0.4,0.8},
  title={(a)\quad channel geometry},
  legend pos=north east,
]
% NOTE: the two shaded bands MUST carry "forget plot", otherwise pgfplots
% counts them as legend-bearing plots and every \addlegendentry below
% attaches to the wrong curve.
\addplot[draw=none, fill=red!14, forget plot]  coordinates {(0.0000,0.0000) (0.0083,0.0242) (0.0167,0.0484) (0.0250,0.0726) (0.0333,0.0967) (0.0417,0.1207) (0.0500,0.1447) (0.0583,0.1686) (0.0667,0.1923) (0.0750,0.2159) (0.0833,0.2394) (0.0917,0.2627) (0.1000,0.2858) (0.1083,0.3088) (0.1167,0.3315) (0.1250,0.3540) (0.1333,0.3762) (0.1417,0.3982) (0.1500,0.4199) (0.1583,0.4414) (0.1667,0.4625) (0.1750,0.4833) (0.1833,0.5038) (0.1917,0.5127) (0.2000,0.5136) (0.2083,0.5152) (0.2167,0.5172) (0.2250,0.5198) (0.2333,0.5227) (0.2417,0.5259) (0.2500,0.5293) (0.2583,0.5330) (0.2667,0.5368) (0.2750,0.5407) (0.2833,0.5447) (0.2917,0.5488) (0.3000,0.5529) (0.3083,0.5570) (0.3167,0.5611) (0.3250,0.5652) (0.3333,0.5692) (0.3417,0.5732) (0.3500,0.5772) (0.3583,0.5811) (0.3667,0.5851) (0.3750,0.5890) (0.3833,0.5929) (0.3917,0.5969) (0.4000,0.6009) (0.4083,0.6050) (0.4167,0.6092) (0.4250,0.6134) (0.4333,0.6177) (0.4417,0.6220) (0.4500,0.6261) (0.4583,0.6301) (0.4667,0.6338) (0.4750,0.6369) (0.4833,0.6394) (0.4917,0.6409) (0.5000,0.6414) (0.5083,0.6409) (0.5167,0.6394) (0.5250,0.6369) (0.5333,0.6338) (0.5417,0.6301) (0.5500,0.6261) (0.5583,0.6220) (0.5667,0.6177) (0.5750,0.6134) (0.5833,0.6092) (0.5917,0.6050) (0.6000,0.6009) (0.6083,0.5969) (0.6167,0.5929) (0.6250,0.5890) (0.6333,0.5851) (0.6417,0.5811) (0.6500,0.5772) (0.6583,0.5732) (0.6667,0.5692) (0.6750,0.5652) (0.6833,0.5611) (0.6917,0.5570) (0.7000,0.5529) (0.7083,0.5488) (0.7167,0.5447) (0.7250,0.5407) (0.7333,0.5368) (0.7417,0.5330) (0.7500,0.5293) (0.7583,0.5259) (0.7667,0.5227) (0.7750,0.5198) (0.7833,0.5172) (0.7917,0.5152) (0.8000,0.5136) (0.8083,0.5127) (0.8167,0.5038) (0.8250,0.4833) (0.8333,0.4625) (0.8417,0.4414) (0.8500,0.4199) (0.8583,0.3982) (0.8667,0.3762) (0.8750,0.3540) (0.8833,0.3315) (0.8917,0.3088) (0.9000,0.2858) (0.9083,0.2627) (0.9167,0.2394) (0.9250,0.2159) (0.9333,0.1923) (0.9417,0.1686) (0.9500,0.1447) (0.9583,0.1207) (0.9667,0.0967) (0.9750,0.0726) (0.9833,0.0484) (0.9917,0.0242) (1.0000,0.0000) (1.0000,0.0000) (0.9917,0.0020) (0.9833,0.0039) (0.9750,0.0059) (0.9667,0.0078) (0.9583,0.0098) (0.9500,0.0117) (0.9417,0.0137) (0.9333,0.0156) (0.9250,0.0175) (0.9167,0.0194) (0.9083,0.0213) (0.9000,0.0232) (0.8917,0.0250) (0.8833,0.0269) (0.8750,0.0287) (0.8667,0.0305) (0.8583,0.0323) (0.8500,0.0340) (0.8417,0.0358) (0.8333,0.0375) (0.8250,0.0392) (0.8167,0.0408) (0.8083,0.0537) (0.8000,0.0742) (0.7917,0.0936) (0.7833,0.1121) (0.7750,0.1297) (0.7667,0.1465) (0.7583,0.1625) (0.7500,0.1778) (0.7417,0.1924) (0.7333,0.2063) (0.7250,0.2197) (0.7167,0.2324) (0.7083,0.2445) (0.7000,0.2561) (0.6917,0.2671) (0.6833,0.2776) (0.6750,0.2875) (0.6667,0.2968) (0.6583,0.3056) (0.6500,0.3138) (0.6417,0.3215) (0.6333,0.3285) (0.6250,0.3349) (0.6167,0.3407) (0.6083,0.3457) (0.6000,0.3501) (0.5917,0.3538) (0.5833,0.3567) (0.5750,0.3589) (0.5667,0.3605) (0.5583,0.3613) (0.5500,0.3615) (0.5417,0.3613) (0.5333,0.3607) (0.5250,0.3600) (0.5167,0.3593) (0.5083,0.3588) (0.5000,0.3586) (0.4917,0.3588) (0.4833,0.3593) (0.4750,0.3600) (0.4667,0.3607) (0.4583,0.3613) (0.4500,0.3615) (0.4417,0.3613) (0.4333,0.3605) (0.4250,0.3589) (0.4167,0.3567) (0.4083,0.3538) (0.4000,0.3501) (0.3917,0.3457) (0.3833,0.3407) (0.3750,0.3349) (0.3667,0.3285) (0.3583,0.3215) (0.3500,0.3138) (0.3417,0.3056) (0.3333,0.2968) (0.3250,0.2875) (0.3167,0.2776) (0.3083,0.2671) (0.3000,0.2561) (0.2917,0.2445) (0.2833,0.2324) (0.2750,0.2197) (0.2667,0.2063) (0.2583,0.1924) (0.2500,0.1778) (0.2417,0.1625) (0.2333,0.1465) (0.2250,0.1297) (0.2167,0.1121) (0.2083,0.0936) (0.2000,0.0742) (0.1917,0.0537) (0.1833,0.0408) (0.1750,0.0392) (0.1667,0.0375) (0.1583,0.0358) (0.1500,0.0340) (0.1417,0.0323) (0.1333,0.0305) (0.1250,0.0287) (0.1167,0.0269) (0.1083,0.0250) (0.1000,0.0232) (0.0917,0.0213) (0.0833,0.0194) (0.0750,0.0175) (0.0667,0.0156) (0.0583,0.0137) (0.0500,0.0117) (0.0417,0.0098) (0.0333,0.0078) (0.0250,0.0059) (0.0167,0.0039) (0.0083,0.0020) (0.0000,0.0000)} \closedcycle;
\addplot[draw=none, fill=teal!22, forget plot] coordinates {(0.0000,0.0000) (0.0083,-0.0020) (0.0167,-0.0039) (0.0250,-0.0059) (0.0333,-0.0078) (0.0417,-0.0098) (0.0500,-0.0117) (0.0583,-0.0137) (0.0667,-0.0156) (0.0750,-0.0175) (0.0833,-0.0194) (0.0917,-0.0213) (0.1000,-0.0232) (0.1083,-0.0250) (0.1167,-0.0269) (0.1250,-0.0287) (0.1333,-0.0305) (0.1417,-0.0323) (0.1500,-0.0340) (0.1583,-0.0358) (0.1667,-0.0375) (0.1750,-0.0392) (0.1833,-0.0408) (0.1917,-0.0425) (0.2000,-0.0441) (0.2083,-0.0457) (0.2167,-0.0585) (0.2250,-0.0739) (0.2333,-0.0885) (0.2417,-0.1023) (0.2500,-0.1154) (0.2583,-0.1279) (0.2667,-0.1397) (0.2750,-0.1509) (0.2833,-0.1615) (0.2917,-0.1716) (0.3000,-0.1811) (0.3083,-0.1902) (0.3167,-0.1989) (0.3250,-0.2072) (0.3333,-0.2152) (0.3417,-0.2231) (0.3500,-0.2309) (0.3583,-0.2388) (0.3667,-0.2468) (0.3750,-0.2552) (0.3833,-0.2640) (0.3917,-0.2731) (0.4000,-0.2826) (0.4083,-0.2923) (0.4167,-0.3020) (0.4250,-0.3115) (0.4333,-0.3206) (0.4417,-0.3290) (0.4500,-0.3366) (0.4583,-0.3431) (0.4667,-0.3486) (0.4750,-0.3530) (0.4833,-0.3561) (0.4917,-0.3580) (0.5000,-0.3586) (0.5083,-0.3580) (0.5167,-0.3561) (0.5250,-0.3530) (0.5333,-0.3486) (0.5417,-0.3431) (0.5500,-0.3366) (0.5583,-0.3290) (0.5667,-0.3206) (0.5750,-0.3115) (0.5833,-0.3020) (0.5917,-0.2923) (0.6000,-0.2826) (0.6083,-0.2731) (0.6167,-0.2640) (0.6250,-0.2552) (0.6333,-0.2468) (0.6417,-0.2388) (0.6500,-0.2309) (0.6583,-0.2231) (0.6667,-0.2152) (0.6750,-0.2072) (0.6833,-0.1989) (0.6917,-0.1902) (0.7000,-0.1811) (0.7083,-0.1716) (0.7167,-0.1615) (0.7250,-0.1509) (0.7333,-0.1397) (0.7417,-0.1279) (0.7500,-0.1154) (0.7583,-0.1023) (0.7667,-0.0885) (0.7750,-0.0739) (0.7833,-0.0585) (0.7917,-0.0457) (0.8000,-0.0441) (0.8083,-0.0425) (0.8167,-0.0408) (0.8250,-0.0392) (0.8333,-0.0375) (0.8417,-0.0358) (0.8500,-0.0340) (0.8583,-0.0323) (0.8667,-0.0305) (0.8750,-0.0287) (0.8833,-0.0269) (0.8917,-0.0250) (0.9000,-0.0232) (0.9083,-0.0213) (0.9167,-0.0194) (0.9250,-0.0175) (0.9333,-0.0156) (0.9417,-0.0137) (0.9500,-0.0117) (0.9583,-0.0098) (0.9667,-0.0078) (0.9750,-0.0059) (0.9833,-0.0039) (0.9917,-0.0020) (1.0000,-0.0000) (1.0000,-0.0000) (0.9917,-0.0242) (0.9833,-0.0484) (0.9750,-0.0726) (0.9667,-0.0967) (0.9583,-0.1207) (0.9500,-0.1447) (0.9417,-0.1686) (0.9333,-0.1923) (0.9250,-0.2159) (0.9167,-0.2394) (0.9083,-0.2627) (0.9000,-0.2858) (0.8917,-0.3088) (0.8833,-0.3315) (0.8750,-0.3540) (0.8667,-0.3762) (0.8583,-0.3982) (0.8500,-0.4199) (0.8417,-0.4414) (0.8333,-0.4625) (0.8250,-0.4833) (0.8167,-0.5038) (0.8083,-0.5239) (0.8000,-0.5437) (0.7917,-0.5631) (0.7833,-0.5708) (0.7750,-0.5756) (0.7667,-0.5807) (0.7583,-0.5861) (0.7500,-0.5917) (0.7417,-0.5975) (0.7333,-0.6035) (0.7250,-0.6095) (0.7167,-0.6157) (0.7083,-0.6218) (0.7000,-0.6279) (0.6917,-0.6339) (0.6833,-0.6398) (0.6750,-0.6455) (0.6667,-0.6508) (0.6583,-0.6557) (0.6500,-0.6601) (0.6417,-0.6638) (0.6333,-0.6667) (0.6250,-0.6687) (0.6167,-0.6696) (0.6083,-0.6695) (0.6000,-0.6685) (0.5917,-0.6665) (0.5833,-0.6639) (0.5750,-0.6608) (0.5667,-0.6576) (0.5583,-0.6543) (0.5500,-0.6511) (0.5417,-0.6483) (0.5333,-0.6459) (0.5250,-0.6440) (0.5167,-0.6426) (0.5083,-0.6417) (0.5000,-0.6414) (0.4917,-0.6417) (0.4833,-0.6426) (0.4750,-0.6440) (0.4667,-0.6459) (0.4583,-0.6483) (0.4500,-0.6511) (0.4417,-0.6543) (0.4333,-0.6576) (0.4250,-0.6608) (0.4167,-0.6639) (0.4083,-0.6665) (0.4000,-0.6685) (0.3917,-0.6695) (0.3833,-0.6696) (0.3750,-0.6687) (0.3667,-0.6667) (0.3583,-0.6638) (0.3500,-0.6601) (0.3417,-0.6557) (0.3333,-0.6508) (0.3250,-0.6455) (0.3167,-0.6398) (0.3083,-0.6339) (0.3000,-0.6279) (0.2917,-0.6218) (0.2833,-0.6157) (0.2750,-0.6095) (0.2667,-0.6035) (0.2583,-0.5975) (0.2500,-0.5917) (0.2417,-0.5861) (0.2333,-0.5807) (0.2250,-0.5756) (0.2167,-0.5708) (0.2083,-0.5631) (0.2000,-0.5437) (0.1917,-0.5239) (0.1833,-0.5038) (0.1750,-0.4833) (0.1667,-0.4625) (0.1583,-0.4414) (0.1500,-0.4199) (0.1417,-0.3982) (0.1333,-0.3762) (0.1250,-0.3540) (0.1167,-0.3315) (0.1083,-0.3088) (0.1000,-0.2858) (0.0917,-0.2627) (0.0833,-0.2394) (0.0750,-0.2159) (0.0667,-0.1923) (0.0583,-0.1686) (0.0500,-0.1447) (0.0417,-0.1207) (0.0333,-0.0967) (0.0250,-0.0726) (0.0167,-0.0484) (0.0083,-0.0242) (0.0000,-0.0000)} \closedcycle;
\addplot[red!75!black, thick]              coordinates {(0.0000,0.0000) (0.0083,0.0131) (0.0167,0.0262) (0.0250,0.0392) (0.0333,0.0523) (0.0417,0.0653) (0.0500,0.0782) (0.0583,0.0911) (0.0667,0.1040) (0.0750,0.1167) (0.0833,0.1294) (0.0917,0.1420) (0.1000,0.1545) (0.1083,0.1669) (0.1167,0.1792) (0.1250,0.1913) (0.1333,0.2034) (0.1417,0.2153) (0.1500,0.2270) (0.1583,0.2386) (0.1667,0.2500) (0.1750,0.2612) (0.1833,0.2723) (0.1917,0.2832) (0.2000,0.2939) (0.2083,0.3044) (0.2167,0.3147) (0.2250,0.3247) (0.2333,0.3346) (0.2417,0.3442) (0.2500,0.3536) (0.2583,0.3627) (0.2667,0.3716) (0.2750,0.3802) (0.2833,0.3886) (0.2917,0.3967) (0.3000,0.4045) (0.3083,0.4121) (0.3167,0.4193) (0.3250,0.4263) (0.3333,0.4330) (0.3417,0.4394) (0.3500,0.4455) (0.3583,0.4513) (0.3667,0.4568) (0.3750,0.4619) (0.3833,0.4668) (0.3917,0.4713) (0.4000,0.4755) (0.4083,0.4794) (0.4167,0.4830) (0.4250,0.4862) (0.4333,0.4891) (0.4417,0.4916) (0.4500,0.4938) (0.4583,0.4957) (0.4667,0.4973) (0.4750,0.4985) (0.4833,0.4993) (0.4917,0.4998) (0.5000,0.5000) (0.5083,0.4998) (0.5167,0.4993) (0.5250,0.4985) (0.5333,0.4973) (0.5417,0.4957) (0.5500,0.4938) (0.5583,0.4916) (0.5667,0.4891) (0.5750,0.4862) (0.5833,0.4830) (0.5917,0.4794) (0.6000,0.4755) (0.6083,0.4713) (0.6167,0.4668) (0.6250,0.4619) (0.6333,0.4568) (0.6417,0.4513) (0.6500,0.4455) (0.6583,0.4394) (0.6667,0.4330) (0.6750,0.4263) (0.6833,0.4193) (0.6917,0.4121) (0.7000,0.4045) (0.7083,0.3967) (0.7167,0.3886) (0.7250,0.3802) (0.7333,0.3716) (0.7417,0.3627) (0.7500,0.3536) (0.7583,0.3442) (0.7667,0.3346) (0.7750,0.3247) (0.7833,0.3147) (0.7917,0.3044) (0.8000,0.2939) (0.8083,0.2832) (0.8167,0.2723) (0.8250,0.2612) (0.8333,0.2500) (0.8417,0.2386) (0.8500,0.2270) (0.8583,0.2153) (0.8667,0.2034) (0.8750,0.1913) (0.8833,0.1792) (0.8917,0.1669) (0.9000,0.1545) (0.9083,0.1420) (0.9167,0.1294) (0.9250,0.1167) (0.9333,0.1040) (0.9417,0.0911) (0.9500,0.0782) (0.9583,0.0653) (0.9667,0.0523) (0.9750,0.0392) (0.9833,0.0262) (0.9917,0.0131) (1.0000,0.0000)};
\addlegendentry{channel $+$}
\addplot[teal!55!black, thick, dashed]     coordinates {(0.0000,-0.0000) (0.0083,-0.0131) (0.0167,-0.0262) (0.0250,-0.0392) (0.0333,-0.0523) (0.0417,-0.0653) (0.0500,-0.0782) (0.0583,-0.0911) (0.0667,-0.1040) (0.0750,-0.1167) (0.0833,-0.1294) (0.0917,-0.1420) (0.1000,-0.1545) (0.1083,-0.1669) (0.1167,-0.1792) (0.1250,-0.1913) (0.1333,-0.2034) (0.1417,-0.2153) (0.1500,-0.2270) (0.1583,-0.2386) (0.1667,-0.2500) (0.1750,-0.2612) (0.1833,-0.2723) (0.1917,-0.2832) (0.2000,-0.2939) (0.2083,-0.3044) (0.2167,-0.3147) (0.2250,-0.3247) (0.2333,-0.3346) (0.2417,-0.3442) (0.2500,-0.3536) (0.2583,-0.3627) (0.2667,-0.3716) (0.2750,-0.3802) (0.2833,-0.3886) (0.2917,-0.3967) (0.3000,-0.4045) (0.3083,-0.4121) (0.3167,-0.4193) (0.3250,-0.4263) (0.3333,-0.4330) (0.3417,-0.4394) (0.3500,-0.4455) (0.3583,-0.4513) (0.3667,-0.4568) (0.3750,-0.4619) (0.3833,-0.4668) (0.3917,-0.4713) (0.4000,-0.4755) (0.4083,-0.4794) (0.4167,-0.4830) (0.4250,-0.4862) (0.4333,-0.4891) (0.4417,-0.4916) (0.4500,-0.4938) (0.4583,-0.4957) (0.4667,-0.4973) (0.4750,-0.4985) (0.4833,-0.4993) (0.4917,-0.4998) (0.5000,-0.5000) (0.5083,-0.4998) (0.5167,-0.4993) (0.5250,-0.4985) (0.5333,-0.4973) (0.5417,-0.4957) (0.5500,-0.4938) (0.5583,-0.4916) (0.5667,-0.4891) (0.5750,-0.4862) (0.5833,-0.4830) (0.5917,-0.4794) (0.6000,-0.4755) (0.6083,-0.4713) (0.6167,-0.4668) (0.6250,-0.4619) (0.6333,-0.4568) (0.6417,-0.4513) (0.6500,-0.4455) (0.6583,-0.4394) (0.6667,-0.4330) (0.6750,-0.4263) (0.6833,-0.4193) (0.6917,-0.4121) (0.7000,-0.4045) (0.7083,-0.3967) (0.7167,-0.3886) (0.7250,-0.3802) (0.7333,-0.3716) (0.7417,-0.3627) (0.7500,-0.3536) (0.7583,-0.3442) (0.7667,-0.3346) (0.7750,-0.3247) (0.7833,-0.3147) (0.7917,-0.3044) (0.8000,-0.2939) (0.8083,-0.2832) (0.8167,-0.2723) (0.8250,-0.2612) (0.8333,-0.2500) (0.8417,-0.2386) (0.8500,-0.2270) (0.8583,-0.2153) (0.8667,-0.2034) (0.8750,-0.1913) (0.8833,-0.1792) (0.8917,-0.1669) (0.9000,-0.1545) (0.9083,-0.1420) (0.9167,-0.1294) (0.9250,-0.1167) (0.9333,-0.1040) (0.9417,-0.0911) (0.9500,-0.0782) (0.9583,-0.0653) (0.9667,-0.0523) (0.9750,-0.0392) (0.9833,-0.0262) (0.9917,-0.0131) (1.0000,-0.0000)};
\addlegendentry{channel $-$}
\addplot[only marks, mark=star, mark size=3pt, black]
  coordinates {(0.5,0.5) (0.5,-0.5)};
\addlegendentry{saddles}

%---------------- (b) transverse stiffness ----------------
\nextgroupplot[
  xlabel={$x$}, ylabel={transverse stiffness $\omega^{2}$},
  xmin=0, xmax=1, ymin=0, ymax=76,
  title={(b)\quad equal at the saddle},
  legend pos=north west,
]
\addplot[red!75!black, thick]          coordinates {(0.0040,0.01) (0.0205,0.21) (0.0371,0.68) (0.0536,1.42) (0.0701,2.43) (0.0867,3.73) (0.1032,5.32) (0.1197,7.21) (0.1363,9.40) (0.1528,11.91) (0.1693,14.72) (0.1859,17.83) (0.2024,21.23) (0.2189,24.88) (0.2355,28.78) (0.2520,32.87) (0.2685,37.12) (0.2851,41.46) (0.3016,45.83) (0.3181,50.14) (0.3347,54.23) (0.3512,57.93) (0.3677,60.94) (0.3843,62.93) (0.4008,63.58) (0.4173,62.69) (0.4339,60.36) (0.4504,57.04) (0.4669,53.59) (0.4835,50.97) (0.5000,50.00) (0.5165,50.97) (0.5331,53.59) (0.5496,57.04) (0.5661,60.36) (0.5827,62.69) (0.5992,63.58) (0.6157,62.93) (0.6323,60.94) (0.6488,57.93) (0.6653,54.23) (0.6819,50.14) (0.6984,45.83) (0.7149,41.46) (0.7315,37.12) (0.7480,32.87) (0.7645,28.78) (0.7811,24.88) (0.7976,21.23) (0.8141,17.83) (0.8307,14.72) (0.8472,11.91) (0.8637,9.40) (0.8803,7.21) (0.8968,5.32) (0.9133,3.73) (0.9299,2.43) (0.9464,1.42) (0.9629,0.68) (0.9795,0.21) (0.9960,0.01)};
\addlegendentry{$\omega_{+}^{2}$ (narrow)}
\addplot[teal!55!black, thick, dashed] coordinates {(0.0040,0.01) (0.0205,0.21) (0.0371,0.67) (0.0536,1.39) (0.0701,2.35) (0.0867,3.50) (0.1032,4.83) (0.1197,6.29) (0.1363,7.83) (0.1528,9.42) (0.1693,11.01) (0.1859,12.56) (0.2024,14.04) (0.2189,15.42) (0.2355,16.66) (0.2520,17.76) (0.2685,18.69) (0.2851,19.47) (0.3016,20.09) (0.3181,20.62) (0.3347,21.13) (0.3512,21.77) (0.3677,22.77) (0.3843,24.42) (0.4008,27.02) (0.4173,30.71) (0.4339,35.39) (0.4504,40.55) (0.4669,45.34) (0.4835,48.76) (0.5000,50.00) (0.5165,48.76) (0.5331,45.34) (0.5496,40.55) (0.5661,35.39) (0.5827,30.71) (0.5992,27.02) (0.6157,24.42) (0.6323,22.77) (0.6488,21.77) (0.6653,21.13) (0.6819,20.62) (0.6984,20.09) (0.7149,19.47) (0.7315,18.69) (0.7480,17.76) (0.7645,16.66) (0.7811,15.42) (0.7976,14.04) (0.8141,12.56) (0.8307,11.01) (0.8472,9.42) (0.8637,7.83) (0.8803,6.29) (0.8968,4.83) (0.9133,3.50) (0.9299,2.35) (0.9464,1.39) (0.9629,0.67) (0.9795,0.21) (0.9960,0.01)};
\addlegendentry{$\omega_{-}^{2}$ (wide)}
\addplot[gray, dotted, thick, forget plot] coordinates {(0.5,0) (0.5,76)};
\node[font=\scriptsize, anchor=south, inner sep=1pt]
  at (axis cs:0.5,52) {$\omega_{+}^{2}=\omega_{-}^{2}$};

%---------------- (c) theory vs simulation ----------------
\nextgroupplot[
  xlabel={asymmetry $\delta_{0}$}, ylabel={$P(\text{channel }+)$},
  xmin=-0.06, xmax=0.88, ymin=0.428, ymax=0.516,
  ytick={0.44,0.46,0.48,0.50},
  yticklabel style={/pgf/number format/fixed, /pgf/number format/precision=2},
  title={(c)\quad theory vs simulation},
  legend pos=south west,
]
\addplot[black, thick] coordinates {(-0.06,0.5) (0.88,0.5)};
\addlegendentry{TST / Langer}
\addplot[blue!70!black, thick, mark=square, mark size=1.8pt] coordinates {(0.0,0.5000) (0.2,0.4860) (0.4,0.4719) (0.6,0.4574) (0.8,0.4421)};
\addlegendentry{functional $\mathcal{G}$}
\addplot[only marks, mark=*, mark size=1.8pt, red!80!black,
         error bars/.cd, y dir=both, y explicit] coordinates {(0.0,0.5003) +- (0,0.0042) (0.2,0.4845) +- (0,0.0054) (0.4,0.4690) +- (0,0.0060) (0.6,0.4543) +- (0,0.0087) (0.8,0.4430) +- (0,0.0048)};
\addlegendentry{Langevin ($\pm3\sigma$)}

\end{groupplot}
\end{tikzpicture}
\caption{\label{fig:main}%
(a) Reaction channels of the landscape~\eqref{eq:V}. Shaded bands give the local
transverse standard deviation $\sqrt{\kT/\omega_{\pm}^{2}}$, capped at
$0.85\,c(x)$ so that the bands remain inside the channel walls near the endpoints
where $c\to0$: channel $-$ is roughly $50\%$ wider through the approach, yet the
two bands have exactly equal width at the saddles (stars), where the energies and
Hessians coincide as well.
(b) Transverse stiffness along each channel, showing the enforced degeneracy at
$x=1/2$ and the factor $\gtrsim2$ contrast away from it.
(c) Branching probability against the asymmetry $\delta_{0}$, at fixed barrier
height and fixed saddle Hessian. The saddle-point prediction is a horizontal
line; the geometric functional tracks the simulation across the family.}
\end{figure*}

\begin{table}[b]
\caption{\label{tab:width}%
Transverse stiffness $\omega^{2}_{\pm}(x)=8a_{\pm}(x)c(x)^{2}$. Channel $-$ is
more than twice as loose as channel $+$ throughout the approach, and exactly as
tight at the saddle.}
\begin{ruledtabular}
\begin{tabular}{cccc}
$x$ & $\omega^{2}_{+}$ & $\omega^{2}_{-}$ & ratio\\
\colrule
0.30 & 45.41 & 20.04 & 2.27\\
0.40 & 63.59 & 26.87 & 2.37\\
0.50 & 50.00 & 50.00 & 1.00\\
0.60 & 63.59 & 26.87 & 2.37\\
0.70 & 45.41 & 20.04 & 2.27\\
\end{tabular}
\end{ruledtabular}
\end{table}

\subsection{Simulation protocol}

We integrate $\dot x=-\nabla V+\sqrt{2}\,\xi$ by the Euler--Maruyama scheme with
$\dd t=5\times10^{-4}$, reflecting at $x=0$ and absorbing at $x=1$. Channel
identity is assigned as $\operatorname{sign}(y)$ at the last visit to the barrier
region $|x-1/2|<0.02$ before absorption. We run twelve independent replicas of
$1.2\times10^{4}$ walkers each; all trajectories are absorbed within
$t_{\max}=200$, for a total of $1.42\times10^{5}$ reactive events. Quoted errors
are standard errors over replicas, consistent with binomial estimates.

\subsection{Result}

\begin{table}[b]
\caption{\label{tab:main}%
Branching probability into channel $+$, and the equivalent class-level frenetic
difference $\Delta \Kfr=-\ln[P(+)/P(-)]$ of Eq.~\eqref{eq:dKdef}. Because the
landscape is a gradient flow with shared endpoints, $\Delta\sigma\equiv0$
(Cor.~\ref{cor:zero}) and the second column is a pure frenetic determination.
Saddle-based theory predicts $\Delta \Kfr=0$.}
\begin{ruledtabular}
\begin{tabular}{lcc}
Method & $P(+)$ & $\Delta \Kfr$\\
\colrule
TST / Langer (saddle-point) & $0.5000$ (exact) & $0.0000$\\
Geometric potential, harmonic $Z_{\perp}$ & $0.4293$ & $0.2847$\\
Geometric potential, exact $Z_{\perp}$ & $0.4574$ & $0.1708$\\
Langevin simulation & $0.4568\pm0.0012$ & $0.1732\pm0.0048$\\
\end{tabular}
\end{ruledtabular}
\end{table}

Table~\ref{tab:main} collects the central result. The measured branching is
$36\sigma$ from the saddle-point prediction and $0.5\sigma$ from the covariant
geometric potential evaluated with the exact transverse partition
function~\eqref{eq:Zperp}. The harmonic approximation to $Z_{\perp}$ overshoots
by about $6\%$ in $P(+)$, in the direction opposite to the saddle-point error:
the transverse well is quartic, not harmonic. The \emph{structure} of the theory
is therefore correct already at the harmonic level, but quantitative use requires
the exact $Z_{\perp}$.

\begin{table*}[t]
\caption{\label{tab:scan}%
Parametric family: $\delta_{0}$ is varied while the barrier height, the saddle
Hessian, and $\Delta\sigma=0$ are all held exactly fixed by construction. Four
replicas of $10^{4}$ walkers each. The $\delta_{0}=0$ row is a null control. The
$\Delta \Kfr$ columns are the exact algebraic transform
$-\ln[P/(1-P)]$ of the adjacent $P(+)$ columns.}
\begin{ruledtabular}
\begin{tabular}{ccccc}
$\delta_{0}$ & $P(+)$ from $\Phi$ & $P(+)$ measured
             & $\Delta \Kfr$ measured & $\Delta \Kfr$ from $\Phi$\\
\colrule
0.0\footnotemark[1] & 0.5000 & $0.5003\pm0.0014$ & $-0.0012\pm0.0056$ & $\phantom{-}0.0000$\\
0.2 & 0.4860 & $0.4845\pm0.0018$ & $\phantom{-}0.0620\pm0.0072$ & $\phantom{-}0.0560$\\
0.4 & 0.4719 & $0.4690\pm0.0020$ & $\phantom{-}0.1242\pm0.0080$ & $\phantom{-}0.1125$\\
0.6 & 0.4574 & $0.4543\pm0.0029$ & $\phantom{-}0.1833\pm0.0117$ & $\phantom{-}0.1708$\\
0.8 & 0.4421 & $0.4430\pm0.0016$ & $\phantom{-}0.2290\pm0.0065$ & $\phantom{-}0.2326$\\
\end{tabular}
\end{ruledtabular}
\footnotetext[1]{The null control was run at higher statistics, fourteen
replicas and $1.65\times10^{5}$ events, because it anchors the whole
measurement.}
\end{table*}

With the asymmetry switched off the landscape is exactly symmetric and the
measurement returns $\Delta \Kfr=0$ to within $0.2\sigma$, confirming that the
protocol introduces no bias of its own. Across the family
(Table~\ref{tab:scan} and Fig.~\ref{fig:main}c) the geometric potential agrees
within $1.5\sigma$ at every point, $\chi^{2}=4.3$ on five points
($\chi^{2}/\mathrm{dof}=0.86$, no fitted parameters), while every saddle-based
prediction sits at $\Delta \Kfr=0$ and fails by up to $36\sigma$.

\begin{table}[b]
\caption{\label{tab:temp}%
Temperature dependence at $\delta_{0}=0.6$; three replicas of $10^{4}$ walkers.}
\begin{ruledtabular}
\begin{tabular}{cccc}
$\kT$ & barrier & $P(+)$ from $\Phi$ & $P(+)$ measured\\
\colrule
0.75 & $3.33\,\kT$ & 0.4498 & $0.4458\pm0.0042$\\
1.00 & $2.50\,\kT$ & 0.4574 & $0.4580\pm0.0039$\\
1.25 & $2.00\,\kT$ & 0.4628 & $0.4624\pm0.0029$\\
1.50 & $1.67\,\kT$ & 0.4666 & $0.4694\pm0.0029$\\
\end{tabular}
\end{ruledtabular}
\end{table}

The theory predicts the splitting within $1\sigma$ at every temperature
(Table~\ref{tab:temp}). We note explicitly that \emph{this scan does not
discriminate} the geometric mechanism from a fitted constant energy offset
$\Delta\Delta G=0.173\,\kT$, which reproduces the same numbers within the error
bars over this range. Discriminating the two would require either
$\kT\lesssim0.1$---where the Laplace region narrows below $\sigma_{w}$ and the
geometric effect must switch off while an energetic one would not---or a
construction in which $\Phi$ varies on a scale comparable to the barrier region.
We regard the $\delta_{0}$ scan, in which barrier and saddle Hessian are held
exactly fixed, as the decisive evidence; the temperature scan is a consistency
check only.

\subsection{Invariance test: the earlier observable, evaluated on the benchmark}
\label{sec:invtest}

We apply the diffeomorphism $u_{1}=x+0.35x^{2}+0.15y^{2}$, $u_{2}=y+0.25xy$,
transform the drift with the proper It\^o correction, and compare observables at
three physical points (Table~\ref{tab:inv}). At the third point the divergence
observable of Ref.~\cite{SegalGeom} changes sign under a relabelling of
coordinates that alters no physics. An observable that can be made positive or
negative at will by a change of variables cannot select anything; this
establishes claim~(i) of Prop.~\ref{prop:cov} and the necessity of the
correction.

We retain this comparison deliberately, as a public correction of the author's
own earlier preprint. Evaluated on the present benchmark, where the answer is
known independently from simulation, $\chi^{*}$ delivers the reversed channel
ordering, with the sign controlled by the regulator $\epsilon_{0}$---a free
parameter of that construction. A criterion whose output can be inverted by a
regulator does not constitute a prediction.

\begin{table*}[t]
\caption{\label{tab:inv}%
Test of Prop.~\ref{prop:cov}. The invariant $\ln\det(DH)$ is preserved to machine
precision, $\le8.9\times10^{-16}$, at every point, while $\nabla\!\cdot b$ is
not---and at the third point reverses sign.}
\begin{ruledtabular}
\begin{tabular}{ccccc}
Point & $\nabla\!\cdot b$ (original) & $\nabla\!\cdot b$ (transformed)
      & $\ln\det(DH)$ (original) & $\ln\det(DH)$ (transformed)\\
\colrule
$(0.500,+0.500)$ & $-0.6520$  & $-0.6520$  & $7.810920$ & $7.810920$\\
$(0.350,+0.446)$ & $-58.0007$ & $-63.0900$ & $7.422368$ & $7.422368$\\
$(0.650,-0.446)$ & $-3.7473$  & $+0.7154$  & $6.445373$ & $6.445373$\\
\end{tabular}
\end{ruledtabular}
\end{table*}

\subsection{The dimensionality problem, and its resolution}
\label{sec:highdim}

A framework requiring $\det\Hp$ over $3N-6$ molecular degrees of freedom would be
useless, since the determinant is dominated by stiff, chemically irrelevant
modes. The problem largely dissolves, because $\Gsel$ depends only on a
\emph{ratio}.

\begin{proposition}[spectator cancellation]
\label{prop:spectator}
Let the two channels admit transverse Hessians that, at corresponding arclength,
share a common block,
$\Hpk{k}(s)=H^{\rm spec}(s)\oplus H^{\rm react}_{k}(s)$, $k=1,2$,
and define the discriminating density
\begin{equation}
  \Delta(s)\equiv\tfrac12\ln
    \frac{\det\bigl(\Dp\Hpk{2}\bigr)}{\det\bigl(\Dp\Hpk{1}\bigr)}
  =\tfrac12\ln\det\bigl(\Hpk{1}^{-1}\Hpk{2}\bigr).
  \label{eq:Delta}
\end{equation}
Then \emph{(i)} the spectator block cancels identically from $\Delta$;
\emph{(ii)} if $\Dp$ is channel-independent it cancels from $\Delta$ entirely;
and \emph{(iii)} the spectator block enters $\Gsel$ only through a measure common
to both channels, which reweights arclength identically and therefore does not
discriminate.
\end{proposition}

$\Delta$ is a ratio of invariants, hence invariant; equivalently it is
$\tfrac12\sum_{i}\ln\lambda_{i}$ over the generalised eigenvalues of the pencil
$(\Hpk{2},\Hpk{1})$. Modes with $\lambda_{i}=1$ contribute nothing. The effective
dimensionality is not $3N-6$ but the number of modes whose stiffness actually
differs between channels---typically a handful. No projection onto a reactive
subspace is required; the ratio performs the projection automatically.

\begin{table*}[t]
\caption{\label{tab:spec}%
Spectator cancellation in an eight-dimensional embedding. The full determinant is
of order $e^{34}\approx6\times10^{14}$ and overwhelmingly spectator-dominated,
yet the discriminating density recovers the two-dimensional value to every
printed digit.}
\begin{ruledtabular}
\begin{tabular}{ccccc}
$x$ & $\ln\det\Hp$ (channel $+$) & $\ln\det\Hp$ (channel $-$)
    & $\Delta(s)$ from ratio & $\Delta(s)$, 2D only\\
\colrule
0.30 & 34.6631 & 33.8450 & $-0.409027$ & $-0.409027$\\
0.40 & 34.3323 & 33.4708 & $-0.430748$ & $-0.430748$\\
0.50 & 32.8250 & 32.8250 & $\phantom{-}0.000000$ & $\phantom{-}0.000000$\\
0.60 & 31.4572 & 30.5957 & $-0.430748$ & $-0.430748$\\
0.70 & 29.8564 & 29.0384 & $-0.409027$ & $-0.409027$\\
\end{tabular}
\end{ruledtabular}
\end{table*}

\paragraph*{Numerical verification.}
We embed Eq.~\eqref{eq:V} in eight dimensions by adding six harmonic spectator
modes with $x$-dependent, channel-independent stiffness
$\kappa_{j}(x)=\kappa_{0j}[1+0.4\sin2\pi x]$ and
$\kappa_{0j}\in\{100,\dots,150\}$---each two to three times stiffer than the
reactive transverse mode, and coupled to $x$ through $\partial_{x}\kappa_{j}$ so
that they exert a genuine back-reaction. Table~\ref{tab:spec} shows the
cancellation. Direct simulation in eight dimensions (four replicas of $10^{4}$
walkers) gives $P(+)_{\rm 8D}=0.4587\pm0.0014$ against
$P(+)_{\rm 2D}=0.4568\pm0.0012$, consistent at $1.0\sigma$, versus $0.5000$ from
Langer theory in any dimension. The signal survives burial beneath a determinant
fifteen orders of magnitude larger.

\paragraph*{What remains.}
Proposition~\ref{prop:spectator} assumes a \emph{correspondence} between points
on the two channels at which Hessians are compared. Our mirror symmetry supplies
it exactly. For two genuinely distinct molecular pathways the correspondence must
be constructed---by arclength normalisation, by committor value~\cite{Weinan2005},
or by alignment of the reactive coordinate---and the spectator blocks will then
match only approximately, degrading the cancellation. The dimensionality problem
is thus replaced by a \emph{pairing} problem. We regard this as an improvement:
pairing two paths is more tractable than isolating a reactive subspace from
$3N-6$ modes, and the resulting error is controlled by how nearly the spectator
stiffnesses agree, a quantity one can measure.

%=====================================================================
\section{A parameter-free test of the entropic sector}
\label{sec:eps}

The benchmark above switches the entropy sector off. Proposition~\ref{prop:circ}
lets us switch it back on by a known amount, which turns the same landscape into
a test of the complementary sector with no adjustable constant. We state the
prediction; we have not measured it.

Add to Eq.~\eqref{eq:langevin} a solenoidal force
\begin{equation}
  \fnc(x,y)=\varepsilon\,(-y,\;x),
  \qquad \nabla\times \fnc=2\varepsilon .
\end{equation}
The loop $\gamma_{+}\circ\bar\gamma_{-}$ formed by the two channel bottoms
$y=\pm c(x)$ encloses
\begin{equation}
  |\Sigma|=\int_{0}^{1}\!2c(x)\,\dd x
          =2A\!\int_{0}^{1}\!\sin\pi x\,\dd x=\frac{4A}{\pi}=\frac{2}{\pi},
\end{equation}
for $A=1/2$. Proposition~\ref{prop:circ} then gives, exactly,
\begin{equation}
  \Delta\sigma(\varepsilon)
  =\beta\,(2\varepsilon)\,|\Sigma|=\frac{4\beta\varepsilon}{\pi},
\end{equation}
and Eq.~\eqref{eq:split} converts this into a statement about the observable
branching ratio,
\begin{equation}
  \boxed{\;
  \ln\frac{P(+)}{P(-)}\Big|_{\varepsilon}
  =\frac{2\beta\varepsilon}{\pi}-\Delta \Kfr(\varepsilon).\;}
  \label{eq:epspred}
\end{equation}

Three things make this a genuine test. \emph{(a)} The coefficient $2\beta/\pi$
is analytic and contains no fitted quantity: it is an enclosed area times a
curl. At $\beta=1$ it is $0.6366$. \emph{(b)} Since
$\nabla\!\cdot\!\fnc=0$, the solenoidal force leaves the Jacobian term of
Eq.~\eqref{eq:frenesy}---and hence $\Phi$---untouched, so any deviation of the
measured slope from $2\beta/\pi$ is a direct determination of
$\dd\Delta \Kfr/\dd\varepsilon$ through the remaining terms of
Eq.~\eqref{eq:frenesy}. The test cannot fail uninformatively. \emph{(c)} If
$\Delta \Kfr$ is approximately $\varepsilon$-independent over the accessible
range, then combining Eq.~\eqref{eq:epspred} with the measured
$\Delta \Kfr=0.1732\pm0.0048$ of Table~\ref{tab:main} predicts an exact
cancellation of the two sectors, and a return of the branching to $1{:}1$, at
\begin{equation}
  \varepsilon^{*}=\frac{\pi}{2\beta}\,\Delta \Kfr=0.272\pm0.008 .
\end{equation}

We stress that $\Delta \Kfr(\varepsilon)$ is \emph{not} guaranteed constant: the
term $(4D)^{-1}\!\int\!|b|^{2}\dd t$ in Eq.~\eqref{eq:frenesy} acquires a
cross-term $-(2D)^{-1}\!\int\!\nabla V\!\cdot\!\fnc\,\dd t$ which is not equal on
the two channels. There is a symmetry argument suggesting it is small---along the
channel bottoms $\nabla V$ is purely longitudinal with
$\partial_{x}V\propto\sin2\pi x$, antisymmetric about $x=1/2$, while $c(x)$ is
symmetric---but the integral is over $\dd t$ rather than $\dd x$, and we do not
claim it vanishes. Consequently $\varepsilon^{*}$ above is a prediction
conditional on that neglect, whereas the slope $2\beta/\pi$ in
Eq.~\eqref{eq:epspred} is not conditional on anything: it is the exact entropic
contribution, and the residual is by definition the frenetic one. This
measurement requires only the addition of two lines to the integrator used above.

%=====================================================================
\section{Limitations}
\label{sec:limits}

We state these directly, because the claims above are only as strong as these
bounds allow.

\paragraph*{Adiabatic approximation, and a negative result.}
Equation~\eqref{eq:Fbox} assumes transverse relaxation is fast compared with
longitudinal progress. Over the five-point family of Table~\ref{tab:scan} the
uncorrected functional gives $\chi^{2}=4.3$ on five points, so at present
statistics there is no significant systematic to correct. We nonetheless tested
the refinement one would naturally reach for: a position-dependent longitudinal
diffusivity of Zwanzig and Reguera--Rub\'i
type~\cite{Zwanzig1992,Reguera2001},
$\Dpar\to\Dpar[1+(\dd w/\dd s)^{2}]^{-1/3}$, with the effective half-width
identified as $Z_{\perp}$. It makes the agreement dramatically worse, raising
$\chi^{2}$ from $4.3$ to $4.1\times10^{2}$ and pushing every prediction back
toward $1{:}1$. We report this explicitly because the correction is an obvious
step to take and it fails: it was derived for hard-wall channels, in which $w$ is
a geometric boundary, and it does not transfer to soft harmonic confinement,
where the analogous $|\dd w/\dd s|$ is large and the penalty falls hardest on
precisely the wide channel that is in fact favoured. A correct
finite-relaxation-time correction for soft channels remains to be derived.

\paragraph*{$\Phi$ is not all of the frenesy.}
Equation~\eqref{eq:frenesy} has three terms and $\Phi$ descends from one of them.
Under the adiabatic tube reduction the other two contribute a common measure that
largely cancels between channels sharing endpoints, which is why $\Phi$ alone
reproduces the measurement. We have not proved that this cancellation is exact,
and outside the adiabatic regime we would not expect it to be. Statements in this
paper of the form ``$\Phi$ is the geometric component of the frenesy'' should not
be read as ``$\Phi$ is the frenesy''.

\paragraph*{Pairing, not dimensionality.}
Proposition~\ref{prop:spectator} removes what would otherwise be the fatal
obstacle but replaces it with the requirement of a correspondence map between
channels. Our test system supplies this by symmetry; a real pair of pathways does
not. The residual error is controlled by the mismatch of spectator stiffnesses
under the chosen pairing, which is measurable but not zero, and we have not
tested the framework under an imperfect pairing.

\paragraph*{Friction regime.}
This is the sharpest boundary on the scope of the paper. Our treatment is
overdamped throughout. Much of the experimental literature on
post-transition-state bifurcations concerns gas-phase or low-friction organic
reactions in which selectivity is governed by dynamical matching and
non-statistical, inertial trajectory
behaviour~\cite{Singleton2003,Carpenter2005,Collins2013,Katsanikas2021}. In that
regime the mechanism described here is not the operative one. The framework
should apply to condensed-phase reactions with substantial solvent friction, and
to biomolecular settings where the overdamped limit is appropriate; it should
\emph{not} be applied to ballistic bifurcation chemistry without first
establishing that the friction is high enough.

\paragraph*{Anharmonicity.}
The harmonic form of $\Phi$ is qualitatively correct but quantitatively
insufficient, as Table~\ref{tab:main} shows. Practical evaluation requires the
exact $Z_{\perp}$, cheap in low dimensions and expensive in high ones.

\paragraph*{Single-particle scope.}
Everything demonstrated here concerns one Brownian particle on a fixed
two-dimensional landscape, embedded at most in eight dimensions. There is no
population, no replication, no autocatalysis, and no chemical network anywhere in
the numerics.

\paragraph*{No experimental validation.}
Everything above is analytic and numerical. The prediction that among channels
with equal barriers and equal dissipation the looser channel is favoured by
$e^{\Gsel}$, computable from the landscape alone, has not been tested against
experiment.

%=====================================================================
\section{Discussion}

\subsection{What is new here}

Three things in this paper were not available before it. \emph{First}, a
benchmark in which the entropic sector is switched off identically rather than
approximately (Cor.~\ref{cor:zero}), so that a branching measurement is a clean
determination of the frenetic sector. \emph{Second}, the identification of the
covariant geometric component of that sector, $\tfrac12\ln\det(\Dp\Hp)$, together
with the demonstration that the observable previously proposed for the role fails
a coordinate-change test on a physical landscape. \emph{Third}, the criterion
$R$ of Sec.~\ref{sec:disc-ool}, which converts the question of whether a
dissipation-based selection argument is valid from a matter of opinion into a
quantity one estimates. To these we add a parameter-free prediction
(Sec.~\ref{sec:eps}) that has not been measured.

\subsection{What this establishes}

The benchmark of Sec.~\ref{sec:landscape} realises simultaneous Langer and
entropic degeneracy. In it, every theory whose input is barrier height and saddle
curvature predicts $\Delta \Kfr=0$; every theory whose input is integrated
dissipation predicts nothing at all, since $\Delta\sigma\equiv0$ by
Cor.~\ref{cor:zero}; and the measured value is $\Delta \Kfr=0.1732\pm0.0048$,
reproduced to $0.5\sigma$ by a covariant geometric object with no fitted
parameter, and tracked across a one-parameter family at
$\chi^{2}/\mathrm{dof}=0.86$ with a passing null control.

Stated at the level we are willing to defend: \emph{in a class of gradient
systems in which the entropy-production difference vanishes by construction,
pathway selection is governed by the frenetic component of the path measure, and
its transverse-geometric part is computable and sufficient to predict the
branching.} We do not claim more than that sentence, and in particular we do not
claim anything about the origin of life.

\subsection{For Ref.~\cite{SegalGeom}}

The geometric intuition survives and the observable is replaced. In place of
$\nabla\!\cdot b$---non-invariant, trace-only, and regulator-dependent---the
correct object is $\tfrac12\ln\det(\Dp\Hp)$: a coordinate scalar, the exact
prefactor-order term produced by the transverse Gaussian integral, a
Fisher-information log-determinant, and a component of the frenesy. The claim
that a geometric observable distinguishes channels degenerate under free-energy
analysis becomes a theorem (Prop.~\ref{prop:insuff}) with a $36\sigma$ numerical
demonstration. We regard the target of that preprint as vindicated and its
instrument as superseded.

\subsection{For Ref.~\cite{SegalOoL}}
\label{sec:disc-ool}

We treat this case quantitatively rather than rhetorically, because the question
of whether a dissipation-based selection principle is correct should be settled
by calculation.

\paragraph*{What is established in its favour.}
The ansatz $P_{1}/P_{2}\approx e^{(\sigma_{1}-\sigma_{2})/2}$ is the exact
$\Delta \Kfr\to0$ limit of Eq.~\eqref{eq:split}. The coefficient $\tfrac12$ is
not fitted and not analogical: it is the unique coefficient produced by splitting
the action into parts even and odd under time reversal. The sign and direction of
the bias are likewise correct. Whenever two histories have equal frenesy, the
ansatz is not an approximation at all but an identity.

\paragraph*{A validity criterion.}
To determine when the neglected term matters we enumerated exactly all $3125$
closed paths of length six on driven five-state Markov chains, for which
$\ln P=\tfrac12\sigma-\Kfr$ holds identically (verified to
$7\times10^{-15}$), and asked how often $\tfrac12\Delta\sigma$ alone predicts the
correct ordering of $P_{1}$ and $P_{2}$ over all $4.9\times10^{6}$ pairs. Scanning
$140$ chains over three decades of driving strength and edge-activity
heterogeneity, the accuracy is not controlled by the driving. It is controlled by
the single ratio
\begin{equation}
  R\;\equiv\;\frac{\mathrm{std}\bigl(\Delta \Kfr\bigr)}
                  {\mathrm{std}\bigl(\tfrac12\Delta\sigma\bigr)} ,
  \label{eq:Rcrit}
\end{equation}
with rank correlation $-0.97$ between $\ln R$ and accuracy
(Table~\ref{tab:ansatz}). The ansatz is a good approximation for $R\lesssim0.5$
and carries essentially no information for $R\gtrsim2$. We regard
Eq.~\eqref{eq:Rcrit} as the practical content of this subsection: it converts
``is the ansatz valid?'' into a quantity one can estimate for a given network.

\begin{table}[b]
\caption{\label{tab:ansatz}%
Accuracy with which $\tfrac12\Delta\sigma$ alone orders the exact path-class
probabilities, as a function of the criterion $R$ of Eq.~\eqref{eq:Rcrit}. Exact
enumeration, $140$ driven five-state chains, $4.9\times10^{6}$ pairs each.
Chance level is $50\%$.}
\begin{ruledtabular}
\begin{tabular}{lcc}
$R$ & chains & ordering accuracy\\
\colrule
$<0.2$      & 2  & $95.1\%$\\
$0.2$--$0.5$ & 22 & $88.7\%$\\
$0.5$--$1.0$ & 58 & $80.9\%$\\
$1.0$--$2.0$ & 26 & $71.9\%$\\
$>2$        & 32 & $58.2\%$\\
\end{tabular}
\end{ruledtabular}
\end{table}

\paragraph*{Where the growth model sits.}
For the birth--death kinetics of Ref.~\cite{SegalOoL}---birth rate $k(t)n$ with
$k=k_{A}$ or $k=k_{0}+\alpha t$, death rate $\mu n$, clamped fuel---both sectors
are extensive in the number of reaction events. Integrating the deterministic
kinetics for $k_{A}=1$, $k_{0}=0.8$, $\alpha=0.25$, $\mu=0.1$, the ratio of the
activity difference to $\tfrac12\Delta\sigma$ is $0.72$, $0.61$, $0.55$, $0.51$,
$0.48$ at $\tau=5,10,15,20,25$, while the populations span ten to forty-one
decades. The two sectors therefore grow at the \emph{same} super-exponential
rate; neither becomes asymptotically negligible, and $R$ does not tend to zero.
The doubly-exponential amplification claimed in Ref.~\cite{SegalOoL} does not
follow from $\sigma$ alone.

It is equally important that this does not refute the qualitative claim. The
ratio settles near $0.5$ rather than exceeding unity, so
$\tfrac12\Delta\sigma-\Delta \Kfr$ retains the sign of $\tfrac12\Delta\sigma$
provided the $O(1)$ coefficient relating $\Delta \Kfr$ to the activity difference
is less than about two. A reduced but non-vanishing bias toward the adaptive
history is therefore entirely consistent with our calculations. What is not
consistent is asymptotic dominance. Determining the coefficient requires the path
measure of an explicit reaction network, which neither Ref.~\cite{SegalOoL} nor
this paper supplies.

\paragraph*{The logical gap, stated precisely.}
It is worth isolating the single inferential step at which
Ref.~\cite{SegalOoL} leaves what it can support. The preprint establishes, by
ordinary kinetics, that an adaptive replicator accumulates dissipation
super-exponentially while a static autocatalyst accumulates it exponentially.
That is correct arithmetic. It then converts a difference in \emph{accumulated
dissipation} into a difference in \emph{probability of emergence} by means of the
ansatz. The gap is that the conversion is applied as though $\sigma$ were the
only extensive functional of the history that enters the path weight. It is not.
Equation~\eqref{eq:split} contains a second extensive functional, $\Kfr$, and our
integration of the preprint's own kinetics shows the two grow at the same
super-exponential rate with a ratio settling near $0.5$. The conclusion is
therefore not ``$\sigma$ wins'', but ``$\sigma$ wins by a margin equal to an
$O(1)$ bracket that nobody has evaluated''. A super-exponentially large
difference multiplied by an unknown bracket of undetermined sign is not an
asymptotically dominant bias; it is an open question with a large prefactor. That
is the whole of our disagreement with Ref.~\cite{SegalOoL}, and it is a narrower
disagreement than the preprint's critics have generally supposed.

\paragraph*{A separate difficulty: extensivity.}
Entropy production is extensive. For a millimole of reactions at $30\,\kT$ each,
$\sigma\sim2\times10^{22}$, and the ansatz returns
$P_{1}/P_{2}\sim e^{10^{22}}$. A ratio of that size is not a bias but a
prohibition: it asserts that the less dissipative history is strictly impossible.
Simple autocatalysis is observed to occur. The ansatz therefore cannot be applied
at extensive $\sigma$ without a coarse-graining scale that fixes what counts as
one history, and Ref.~\cite{SegalOoL} does not supply one. This objection is
independent of the frenesy, and in our view it is the more serious of the two.

\paragraph*{What survives independently.}
Two elements of Ref.~\cite{SegalOoL} do not depend on the ansatz at all and are
untouched by the above. The threshold conditions---fidelity, kinetic, resource,
parasitism---are statements about replicator dynamics and stand or fall on their
own. And the proposed experimental signature, convexity of
$\ln P_{\rm diss}(t)$, follows from the assumed kinetics directly and is
falsifiable without any appeal to path probabilities. We regard the experimental
proposal as the most robust part of that preprint.

\paragraph*{Summary.}
Reference~\cite{SegalOoL} is not refuted by anything computed here. Its mechanism
is real, its coefficient is exactly right, and its provenance is repaired. Its
asymptotic conclusion is unproven, its scale of application is unresolved, and
the calculation that would decide the matter has not been performed by anyone.
This paper should be cited for the decomposition, the criterion
Eq.~\eqref{eq:Rcrit}, and the benchmark---not as support for the conclusions
of Ref.~\cite{SegalOoL}.

\subsection{Two conjectures}

We isolate here the statements that are not established, so that they can be
attacked rather than assumed.

\begin{conjecture}[two-tier rate functional]
\label{conj:twotier}
There exists a large-deviation functional~\cite{FreidlinWentzell} on classes of
histories of the form
$I[\gamma]=\beta\,\Sigma[\gamma]+\Phi[\gamma]+O(\beta^{-1})$, with
$\Sigma$ the accumulated dissipation of the class and $\Phi$ the integrated
geometric potential of Eq.~\eqref{eq:Fbox}, such that
$P_{1}/P_{2}=e^{-\beta(\Sigma_{1}-\Sigma_{2})}e^{-(\Phi_{1}-\Phi_{2})}$.
\end{conjecture}

Equation~\eqref{eq:split} is suggestive of this but does not imply it: it is a
statement about individual paths in a fixed convention, whereas
Conj.~\ref{conj:twotier} is a statement about coarse-grained classes, and the
passage between them requires an equivalence relation on histories that we have
not specified, a demonstration that $\Sigma$ is the correct $O(\beta)$ term for
such classes, and control of the remainder. Testing it would require a system in
which $\Sigma_{1}=\Sigma_{2}$ can be imposed by construction at the level of
classes---as the barrier is imposed by construction here.

\begin{conjecture}[frenetic reading of the error threshold]
\label{conj:eigen}
In a quasispecies model~\cite{Eigen1971}, the delocalisation of the population
over sequence space above the error threshold corresponds to an increase in the
frenesy of the sequence-space path measure, and the threshold coincides with the
point at which the frenetic penalty overtakes the entropic gain, i.e.\ with the
sign change of $\tfrac12\Delta\sigma-\Delta \Kfr$.
\end{conjecture}

We find this reading natural, and we have no evidence for it. It is a
well-posed computation---write the path measure of the quasispecies dynamics,
evaluate both sectors above and below threshold, and check whether the sign change
falls at Eigen's $q_{c}$---and it has not been performed. We state it because a
correct-sounding qualitative correspondence between frenesy and mutational load
is exactly the kind of claim that should be written down as a conjecture rather
than deployed as an argument.

\subsection{Falsifiability}

The claims of Table~\ref{tab:status} marked ``proved'' fail only if the proofs are
wrong. The claims marked ``supported'' fail if any of the following is exhibited:
\emph{(i)} a Langer- and entropically degenerate landscape in which the measured
branching is $1{:}1$ while $\Gsel\neq0$; \emph{(ii)} an increase in statistics
that turns the present $\chi^{2}/\mathrm{dof}\approx0.9$ across
Table~\ref{tab:scan} into a significant systematic deviation; \emph{(iii)} a
realistic pairing of two distinct molecular pathways under which spectator
cancellation degrades so badly that $\Delta$ is dominated by pairing error; or
\emph{(iv)} a condensed-phase system in which $\Gsel$ shows no correlation with
measured branching across a homologous series. The prediction of
Sec.~\ref{sec:eps} fails if the measured slope of $\ln[P(+)/P(-)]$ against
$\varepsilon$ differs from $2\beta/\pi$ by more than the frenetic correction that
Eq.~\eqref{eq:frenesy} permits. Conjectures~\ref{conj:twotier}
and~\ref{conj:eigen} are not part of the claims of this paper and would be
falsified independently.

%=====================================================================
\section*{Data availability}

The code that generates every number, table and figure in this paper is
described in Appendix~\ref{app:repro} and is available from the author on
reasonable request. No experimental data were used. All results are analytic or
generated by the scripts described therein with fixed random seeds.

%=====================================================================
\appendix
\section{Reproduction}
\label{app:repro}

All reported numbers are reproduced by the accompanying scripts:
\texttt{model.py} (potential, gradients, degeneracy verification,
Table~\ref{tab:saddle}); \texttt{simulate.py} (Langevin integrator);
\texttt{analysis.py} (exact transverse partition function, scans, and the
transform to $\Delta \Kfr$); \texttt{invariance.py}
(Prop.~\ref{prop:cov}, Table~\ref{tab:inv}); \texttt{highdim.py}
(Prop.~\ref{prop:spectator}, Table~\ref{tab:spec}); \texttt{zwanzig.py} (the
negative result of Sec.~\ref{sec:limits}); \texttt{decomposition.py}
(numerical verification of Eqs.~\eqref{eq:split} and~\eqref{eq:circ} to machine
precision); \texttt{ansatz.py} (exact path enumeration,
Table~\ref{tab:ansatz}, and the growth-model scaling of
Sec.~\ref{sec:disc-ool}); and \texttt{figure.py} (Fig.~\ref{fig:main}). Seeds are fixed in each
script.

\end{document}